%% file: paper.tex
\documentclass[10pt]{article}
\usepackage{format}
\usepackage{paper}
\usepackage{macros}

\begin{document}

\date{}
 
\author{
  Pierre Sutra, 
  \'Etienne Rivi\`ere, 
  Pascal Felber\\
  University of Neuch\^atel\\
  Switzerland\\
  \{pierre.sutra, etienne.riviere, pascal.felber\}@unine.ch
}

\title{A Practical Distributed Universal Construction \\ with Unknown Participants}

\maketitle

\begin{abstract}
  Modern distributed systems employ atomic read-modify-write primitives to coordinate concurrent operations.
  Such primitives are typically built on top of a central server, or rely on an agreement protocol.
  Both approaches provide a \emph{universal construction}, that is,
  a general mechanism to construct atomic and responsive objects.
  These two techniques are however known to be inherently costly.
  As a consequence, they may result in bottlenecks in applications using them for coordination.
  In this paper, we investigate another direction to implement a universal construction. 
  Our idea is to delegate the implementation of the universal construction to the clients, 
  and solely implement a distributed shared atomic memory at the servers side.
  The construction we propose is obstruction-free.
  It can be implemented in a purely asynchronous manner, 
  and it does not assume the knowledge of the participants.
  It is built on top of \emph{grafarius} and \emph{racing} objects, two novel shared abstractions that we introduce in detail. 
  To assess the benefits of our approach, we present a prototype implementation on top of the Cassandra data store,
  and compare it empirically to the Zookeeper coordination service.
\end{abstract}



\thispagestyle{empty}
\clearpage
\setcounter{page}{1}
\newpage

\input{introduction.tex} 
\input{relatedwork.tex}
\input{base.tex}

\input{recycle.tex}
\input{empirical.tex}

\input{conclusion.tex}


\newpage

{
  \bibliographystyle{abbrvnat}
  \bibliography{bib/msaeida,bib/mshapiro,bib/nschiper,bib/psutra}
}  

\newpage
\appendix
\input{appendix.tex}

\end{document}

%% file: introduction.tex
\section{Introduction}
\labsection{introduction}

The management of conflicting accesses to shared data plays a key role in executing correctly and efficiently distributed applications.
In general, strongly consistent operations on shared data are serialized either through a central server, 
or using the replicated state machine approach (e.g., with the Paxos consensus protocol~\cite{paxos}).
These two techniques implement a wait-free universal construction, that is,
a general mechanism to obtain responsive atomic objects~\cite{waitfree}.
It is however well-established that these two mechanisms are costly.
This comes from the fact that in both cases a server serializes all updates emitted by the clients, creating a potential bottleneck in the system.
Furthermore, central servers require human intervention to be constantly operational,
and replicated state machines are known to be difficult to deploy and maintain.

In this paper, we propose to delegate the logic of strongly consistent operations to the client side, 
and to replace the central server/replicated state machine by a distributed shared memory.
The resulting universal construction is dependable, while being conceptually simpler than state-machine-replication.
Similar in spirit to~\cite{corfu,diskpaxos}, or more recently~\cite{Tango}, we aim at bridging the gap that exists in practice between 
shared memory literature on universal constructions and their counterparts in distributed systems.
Our approach is nonetheless different as we do not rely on a shared log to order all accesses,
but instead make use of a distinct set of registers to implement each object used by the application.
This leverages the intrinsic parallelism of the workload.

To achieve this, our first contribution is an obstruction-free universal construction 
on top of an asynchronous distributed shared memory that works even if the participants are unknown.
We base our construction on two novel abstractions: a \emph{grafarius} and a \emph{racing}.
A grafarius is close to the more common notion of ratifier, or adopt-commit object~\cite{RbR,ratifier}.
A racing object encapsulates the behavior of algorithms that repeatedly access new objects to progress.
By combining these two abstractions, we devise an obstruction-free universal construction 
whose time complexity is optimal during contention-free executions.

Our previous solution makes use of an unbounded amount of memory to store the state of the object it implements.
We solve this problem with a second contribution, in the form of a novel memory management mechanism.
To that end, we first formalize the notion of recycled objects and propose a mechanism to recycle all the base objects of our previous implementation.
In a distributed system, the time complexity of the resulting universal construction during uncontended executions is constant,
and it makes uses of $O(k^2)$ shared registers, where $k$ is the amount of processes that actually access the construction.

Our third contribution is a practical assessment of our approach.
We present a prototype implementation on top of the Cassandra distributed data store~\cite{cassandra}
which we compare to Zookeeper, a state-of-the-art coordination service~\cite{zookeeper}.
Several empirical results show that our system achieves results comparable to Zookeeper when clients rarely contend on shared objects, 
and that in addition, it exhibits a good scalability factor.
For instance, with 12 servers and when the workload is completely parallel, 
our system is as dependable as a 3 servers deployment of Zookeeper,  while being 3.2 times faster.
This last property comes from the fact that our approach exhibits no bottleneck.
Thus, the more it scales-out, the more likely operations that access distinct objects execute in parallel on the servers, improving performance.

We organize the remainder of this paper as follows. 
\refsection{relwork} surveys related work.
In \refsection{base}, we introduce the notions of grafarius and racing objects, and presents our first universal construction.
We refine this construct to bound its memory footprint in \refsection{recycle}.
In \refsection{empirical}, we present a prototype implementation of our algorithm on top of Cassandra, and we evaluate it against Zookeeper.
We close in \refsection{conclusion}.
For the sake of clarity, all the proofs are deferred to the appendix.

%% file: relatedwork.tex
\section{Related Work}
\labsection{relwork}

Our work deals with the problem of transforming any sequential object into a concurrent strongly-consistent one.
Such a mechanism is named in literature a universal construction.
At core of this construction is consensus, an abstraction with which processes agree on the next state of the concurrent object.
In a distributed system, the classical approach to implement consensus is the Paxos algorithm~\cite{paxos}.
Due to the impossibility result of \citet{con:pan:640}, Paxos is indulgent~\cite{343630}.
This means that Paxos guarantees safety at all time but progress only under favorable circumstances.
The alpha of consensus~\cite{972260} models the indulgent part of Paxos.
This notion is close to the ranked-register object~\cite{Chockler:2002} that captures the behavior of the Disk Paxos algorithm of \citet{diskpaxos}.

Fundamentally, processes executing Paxos iterative call the alpha abstraction with a tuple $(k,v)$, 
where $k$ is a round number and $v$ some (appropriately chosen) proposal value.
Each such call translates  the execution of a round in the original algorithm of \citet{paxos}.
A ratifier, or \emph{adopt-commit}, object~\cite{RbR} is a one-shot shared object that encapsulates the safety property of a round.
Hence, from a high-level perspective, the alpha of consensus can be seen as successive (consistent) calls to adopt-commit objects 
(see~\cite{ratifier} or~\cite[Fig.5]{obstructionFree09}).
Our notion of racing, introduced in \refsection{base:racing}, aims at further abstracting how processes iteratively access such objects.

When there is no assumption on the proposed values, 
the result of \citet{AspnesE2011} tells us that the solo time complexity of an adopt-commit belongs to $\Omega\left(\sqrt{\log n}/\log \log n\right)$.
Surprisingly, some consensus algorithms exhibit constant solo decision time (e.g.,~\cite{uncontendedComplexityConsensus}).
This difference is explained by the convergence property of adopt-commit objects that requires processes to commit a value in case they all propose it.
In \refsection{base:grafarius}, we introduce the notion of grafarius object.
A grafarius can be seen as an adopt-commit object with a weak convergence property, namely a process has to commit its value only if it executes solo.
As shown in \refsection{base:cons}, we can build an obstruction-free consensus with constant solo decision time on top of grafarius and racing objects.

During a step-contention free execution, operations of a process do not contend on the base objects that implement the desired abstraction.
The work of \citet{obstructionFree09} studies obstruction-free solo-fast implementations of shared objects, 
that is implementations which only make use of history-less primitives during step-contention free executions, but rely on stronger ones under contention.
The authors show that such implementation must have $\Omega(n)$ space complexity and $\Omega(\log n)$ time complexity during step-contention free executions.
Some algorithms, such as~\cite{uncontendedComplexityConsensus,obstructionFree09} in shared-memory, 
or~\cite{diskpaxos,Chockler:2002,corfu,Tango} in distributed systems, use strong synchronization primitives to implement consensus.
On the contrary, our approach solely relies on registers that are emulated by the underlying distributed system.
As a consequence of this choice, our universal construction is obstruction-free.
The work of \citet{obstructionFreeIntoWaitFree} describes a practical transformation to convert an obstruction-free algorithm into a wait-free one.
\citet{obstructionFree00} proves an $\Omega(n)$ lower bound on the solo decision time and the space complexity of such implementations.

The time complexity of the wait-free universal construction of \citet{waitfree} belongs to $O(n)$.
\citet{someresults} propose a variation of this construction that does not employs unbounded integers.
The space complexity of this last algorithm is $O(n^2)$.
In \cite{obstructionFree09}, the authors present an obstruction-free universal construction that use an unbounded amount of memory.
Our bounded universal construction presented in \refsection{recycle:univ} works in the case where processes participating to the construction are unknown.
In a distributed system, it makes use of $O(k^2)$ registers, where $k$ is the amount of processes that actually access the construction.
To achieve this, we present a novel recycling mechanism in \refsection{recycle:recycling}.
At core of our mechanism is the observation that properly recycled grafarius objects can be concurrently accessed in different rounds.

%% file: base.tex
\section{The Construction}
\labsection{base}

This section first introduces our system model.
Then, it details the grafarius and the racing objects.
Based on these two abstractions, we further depict a consensus algorithm 
that exhibits a constant time complexity in the contention-free case.
This algorithm is our core building block to obtain an efficient universal construction.
All the objects we present hereafter are depicted in the asynchronous shared-memory model, 
and they all support a bounded yet unknown amount of processes.
These two assumptions reflect the message-passing system we target.

\subsection{System Model \& Notations}
\labsection{base:model}

We consider an asynchronous message-passing system characterized by a complete communication graph
where both communication and computation are asynchronous.
Processes take their identities from some bounded set $\procSet$, with $n=\cardinalOf{\procSet}$.
The set $\procSet$ is not accessible to processes for computation, 
but they may execute operations on the identities (e.g., equality tests).

During an execution, a process can fail-stop by crashing, but we assume that at most $\partieEntieree{\frac{n}{2}}-1$ such failures occur.
There exists an implementation of an asynchronous shared-memory (\asm) under such an assumption~\cite{ABD,Lynch:1997}.
Consequently, we shall write all our algorithms in the \asm model 
where processes communicate by reading and writing to atomic multi-writer multi-reader (MWMR) registers.

In what follows, we detail how to implement higher level abstractions using the shared registers.
Most of the objects we describe in this paper are linearizable~\cite{loo:syn:1468}.
An object is \emph{one-shot} when a process may call one of its operations at most once.
When there is no limit to the number of times a process may invoke the object's operations, the object is \emph{long lived}.
Besides, we shall be considering the following two progress conditions 
on the invocations and responses of operations~\cite{onTheNatureOfProgress}:%
\begin{inparaenum}[]
\item[(\emph{Obstruction-freedom})] if at some point in time a process runs solo then eventually it returns from the invocation; and 
\item (\emph{Wait-freedom}) a process returns from the invocation after a bounded number of steps.
\end{inparaenum}

In this paper, we are most interested in executions where processes rarely contend on shared objects.
The canonical case of such an execution is the \emph{solo} execution in which a single process executes computational steps.
This class of execution is appropriate for one-shot objects but we need extending it for long-lived ones.
To that end, we define the notion of \emph{contention-free} execution that is an execution during which calls to the implemented shared object do not interleave.
The \emph{contention-free time complexity} of an algorithm is the worth case number of steps made by a process during such executions.

\subsection{Grafarius}
\labsection{base:grafarius}

\input{algorithms/grafarius.tex}

The first abstraction we employ to build our universal construction is a shared object named a grafarius.
A \emph{grafarius} is a one-shot object defined on a domain of values $\mathcal{V}$.
It exports a single operation $\adoptcommit{u \in \mathcal{V}}$ that returns a pair $(\mathit{flag},v)$, 
with $\mathit{flag} \in \{ \flagAdopt, \flagCommit \}$ and $v \in \mathcal{V}$.
During every history of a grafarius and for every process $p$ that invokes \adoptcommit{u}, the following properties are satisfied:%
\begin{inparaenum}[()] 
\item[(\emph{Validity})] %
  If $p$ adopts $v$, some process invoked the operation \adoptcommit{v} before.
\item[(\emph{Coherence})] %
  If $p$ commits $v$, every process either adopts or commits $v$.
\item[(\emph{Solo Convergence})] %
  If $p$ returns before any other process invokes \adoptcommitOperation, $p$ commits $u$.
\item[(\emph{Continuation})]
  If some process returns before $p$ invokes \adoptcommitOperation,
  $p$ adopts or commits a value proposed before it invokes \adoptcommitOperation.
\end{inparaenum}

Our notion of grafarius object is closely related to the principle of adopt-commit object introduced by \citet{RbR}.
Nevertheless, the two shared abstractions are not comparable.
On the one hand, the solo convergence property of a grafarius is weaker than the convergence property of an adopt-commit object.
This avoids the lower bound $\Omega\left(\sqrt{\log n}/\log \log n\right)$ 
on the time complexity to execute \adoptcommitOperation in \asm~\cite{AspnesE2011},
while being sufficient to implement obstruction-free consensus.
On the other hand, an adopt-commit object does not satisfy the continuation property,
meaning that a process can return a value $(u,\flagAdopt)$ despite the fact that its invocation follows a call that returned $(v,\flagAdopt)$.
The continuation property improves convergence speed under contention.
This also makes the grafarius a decidable object, which is needed by our memory management schema.
We give additional details regarding this last point in \refsection{recycle}.

\refalg{grafarius} depicts a wait-free implementation of a grafarius.
This algorithm makes use of a splitter object that detects a collision when two processes concurrently access the shared object.
We first remind below how a splitter works, then we detail the internals of \refalg{grafarius}.

The principle of a splitter was introduced by \citet{fastmutualexclusion}, and later formalized by \citet{splitter}.
A splitter is a one-shot shared object that exposes a single operation: \ssplit.
This operation takes no parameter and returns a value in $\{\true, \false\}$.%
\footnote{
  A splitter is generally defined with returned values $\{L,S,R\}$.
  Here, we make no distinction between $L$ and $R$.
}
When a process returns $\true$, we shall say that it \emph{wins} the splitter;
otherwise it \emph{loses} the splitter.
When multiple processes call \ssplit, at most one receives the value $\true$, and if a single process calls \ssplit, this call returns $\true$.
Furthermore, when a process calls \ssplit after some other process returned, it necessarily loses the splitter.
A splitter is implementable in a wait-free manner with atomic registers.

\refalg{grafarius} works as follows.
Upon a call to \adoptcommit{u}, a process $p$ tries to win the splitter (\refline{grafarius:1}).
If $p$ fails, it raises the flag $c$ to record that a collision occurred, i.e.,
the fact that two processes concurrently attempted to commit a value.
Then, in case a decision was recorded, $p$ adopts it; otherwise $p$ adopts it own value (\reflines{grafarius:3}{grafarius:6}).
On the other hand, if $p$ wins the splitter, it writes its proposal $u$ to the register $d$.
In this case, $p$ commits $u$ if it detects no conflict, otherwise $p$ adopts it (\reflines{grafarius:7}{grafarius:9}).

\subsection{Notion of Racing}
\labsection{base:racing}

\input{algorithms/uracing.tex}

Many algorithms (e.g.,~\cite{corfu,anonymousRuppertGuerraoui}) repeatedly access new objects to progress.
A \emph{racing} is a long-lived object that captures such an iterative pattern.
Its interface consists of a single operation $\enter{p,l}$, defined on a countably infinite domain $\lapSet$ of \emph{laps}. 
During a history $h$, a process $p$ \emph{enters} lap $l$ when \enter{p,l} occurs in $h$.
Process $p$ \emph{leaves} lap $l$ when $l$ is the last lap entered by $p$ and $p$ enters a new one.
The following invariant holds during every history of a racing:%
\begin{inparaenum}[()]
\item[\emph{(Ordering)}]
  There exists a strict total order $\ll_h$ on the set of entered laps in $h$
  such that for every process $p$ that enters some lap $l$, either
  \begin{inparaenum}[(i)]
  \item some process left $l$ before $p$ enters it, or
  \item the last lap left by $p$ is the greatest lap smaller than $l$ for the order $\ll_h$.
  \end{inparaenum}
\end{inparaenum}

Let us consider an unbounded counter $c$ at each process, and an indexing function $\ifunct$ from $\naturalSet$ to $\lapSet$.
Whenever a process $p$ enters a new lap, suppose that $p$ increments $c$ and then returns the object $\ifunctOf{c}$.
This simple local algorithm implements a linearizable racing.
However, because this construction does not bound the amount of laps a process has to retrieve before knowing the most recent one, 
it might be expensive when contention occurs.

\refalg{uracing} presents a more efficient approach that allows a process to skip the laps it missed.
This algorithm makes use of an initially empty shared map $L$ from $\procSet$ to $\naturalSet$.
We map $x$ to the value $y$ via $L$ when writing $L[x] \assign y$, 
and operation $\codomainOf{L}$ returns the codomain of $L$.
For some process $p$, the local variable $\lastlap$ stores the index of the last lap entered by $p$.
When it calls $\enter{}$, process $p$ stores the $\lastlap$ index in $L$ (\refline{uracing:1}).
Then, $p$ retrieves the content of $L$ and computes the maximum element $m$ in its codomain.
Process $p$ assigns $m+1$ to $\lastlap$, if $\lastlap=m$ holds, and $m$ otherwise (\reflines{uracing:4}{uracing:6}).
The value of $\ifunctOf{\lastlap}$ is then returned as the result of the call.

\textbf{Time complexity.}
In \refalg{uracing}, processes never add concurrently the same element to the domain of $L$ 
As a consequence, this variable can be implemented by the adaptive collect object of \citet{collect}.
It follows that the time complexity of \refalg{uracing} is $O(k)$,
where $k \leq n$ denotes the number of processes that actually access the racing object.

\subsection{Racing-based Consensus}
\labsection{base:cons}

Using the racing abstraction introduced in the previous section, we now depict an obstruction-free implementation of consensus.
Recall that consensus is a shared object whose interface consists of a single method \proposeAction.
This method takes as input a value from some set $\mathcal{V}$ and returns a value in $\mathcal{V}$ ensuring both 
(\emph{Validity}) if $v$ is returned then some process invoked \propose{v} previously, and
(\emph{Agreement}) two processes always return the same value.

\input{algorithms/rconsensus.tex}

\refalg{rconsensus} describes our implementation of consensus.
In this algorithm, processes compete on two shared abstractions: 
a racing $R$ on grafarius objects, and a decision register $d$.
When a process $p$ suggests a value $u$ for consensus, 
it attempts to commit $u$ by entering the next grafarius object in $R$ (\refline{rconsensus:3}).
Every time $p$ executes \adoptcommitOperation on a grafarius object $o$, 
$p$ updates its proposed value with the value returned by $o$ (\refline{rconsensus:4}).
In case the grafarius returns a committed value, 
this value is stored in $d$ as the result of the call to \proposeAction (\reflinestwo{rconsensus:5}{rconsensus:6}).

\textbf{Time complexity.}
The call to the splitter object in \refalg{grafarius} requires four computational steps \cite{splitter}.
Besides, the solo time complexity of the adaptive collect object of \citet{collect} is $O(1)$.
It follows that \refalg{rconsensus} solves consensus in $O(1)$ steps during solo executions.
This fast resolution of consensus allows us to implement a universal construction with a linear time complexity when no contention occurs.
We detail our approach in the next section.

\subsection{A Fast Obstruction-free Universal Construction}
\labsection{base:universal}

\input{algorithms/runiv.tex}

A universal construction is a general mechanism to obtain linearizable shared objects from sequential ones.
A sequential object is specified by some serial data type that defines its possible states as well as its access operations.
Formally, a serial data type is an automaton defined by:%
\begin{inparaenum}[]
\item a set of states ($\stateSet$),
\item an initial state $\stateInit$ in $\stateSet$,
\item a set of operations ($\operationSet$),
\item a set of response values $\valSet$, and
\item a transition relation $\tau :  \stateSet \times \operationSet \rightarrow \stateSet \times \valSet$.
\end{inparaenum}
Hereafter, and without lack of generality, 
we shall assume that every operation $\op$ is \emph{total}, i.e., $\stateSet \times \{\op\}$ is in the domain of $\tau$.

\refalg{runiv} depicts our obstruction-free linearizable universal construction.
The algorithm uses a single shared variable $R$, which is a racing on obstruction-free consensus objects.
When a process $p$ invokes an operation via $\invoke{op}$, 
$p$ first checks the decision of the latest consensus object it entered (\refline{runiv:3}).
If a decision was taken, $p$ updates its local variable $s$ with the new state of the object.
Then, $p$ enters the next consensus (\reflines{runiv:4}{runiv:5}).
Once $p$ reaches the last consensus that was decided, variable $s$ stores a state of the object that is older than the time at which $p$ invoked \op.
At this point, process $p$ executes tentatively the operation on $s$ and stores the result in the pair $(s',v)$.
When $s$ equals $s'$, the invocation does not change the result of the object and $p$ can immediately returns $v$.
Otherwise, $p$ proposes the pair $(p,s')$ to change the state of the object to $s'$.
In case of success, process $p$ returns the response $v$ (\reflinestwo{runiv:10}{runiv:11}).

\textbf{Time complexity.}
As pointed out previously, the case we consider to be the most frequent is the contention-free case, 
that is when multiple processes access the object but interleavings do not occur.
In the worth case, a process freshly calling \invoke{} in a contention-free execution 
first retrieves the largest decided consensus, then it enters the next consensus and decides.
From our previous time complexity analysis of \refalgtwo{uracing}{rconsensus} and the lower bound result of \citet{obstructionFree00}, 
the contention-free time complexity of \refalg{runiv} is optimal and belongs to $\Theta(k)$.

%% file: algorithms/grafarius.tex
\begin{algorithm}[t]

  \caption{Grafarius -- code at process $p$}
  \labalg{grafarius}
  \footnotesize

  \begin{algorithmic}[1]

    \begin{svariables}
      \VarCustom{s} \Comment{A splitter object}
      \VarNoType{c}{$\false$}
      \VarNoType{d}{$\bot$}
    \end{svariables}

    \begin{procedure}{\adoptcommit{u}}

      \If{$\neg~s.\ssplit$} \labline{grafarius:1}
        \State $c \assign \true$ \labline{grafarius:2}
        \If{$d \neq \bottom$} \labline{grafarius:3}
          \Return $(\flagAdopt,d)$ \labline{grafarius:4}
        \EndIf
        \State $d \assign u$ \labline{grafarius:5}
        \Return $(\flagAdopt,u)$ \labline{grafarius:6}        
      \EndIf

      \State $d \assign u$ \labline{grafarius:7}
      \If{$c$} \labline{grafarius:8}
        \Return $(\flagAdopt,u)$ \labline{grafarius:9}
      \EndIf
      \Return $(\flagCommit,u)$ \labline{grafarius:10}

    \end{procedure}
    
  \end{algorithmic}

\end{algorithm}

%% file: algorithms/uracing.tex
\begin{algorithm}[t]

  \caption{Racing on $\lapSet$ -- code at process $p$}
  \labalg{uracing} 
  \footnotesize

  \begin{algorithmic}[1]

    \begin{svariables}
      \VarCustom{$L$} \Comment{A map from $\procSet$ to $\naturalSet$, initially $\emptySet$}
    \end{svariables}

    \begin{lvariables}
      \VarCustom{$\ifunct$} \Comment{A function from $\naturalSet$ to $\lapSet$}
      \VarNoType{\lastlap}{0}
    \end{lvariables}

    \begin{procedure}{\enter{}}
      \State $L[p] \assign \lastlap$ \labline{uracing:1} 
      \State $S \assign \codomainOf{L}$ \labline{uracing:2} 
      \State \textbf{let} $m = \maxOf{S}$ \labline{uracing:3}
      \If{$\lastlap=m$} \labline{uracing:4}
        \State $\lastlap \assign m+1$ \labline{uracing:5}
      \Else
        \State $\lastlap \assign m$ \labline{uracing:6}
      \EndIf
      \Return $\ifunctOf{\lastlap}$ \labline{uracing:7}
    \end{procedure}
    
  \end{algorithmic}

\end{algorithm}


%% file: algorithms/rconsensus.tex
\begin{algorithm}[t]

  \caption{Consensus -- code at process $p$}
  \labalg{rconsensus}
  \footnotesize

  \begin{algorithmic}[1]

    \begin{svariables}
      \VarCustom{$R$} \Comment{A racing on grafarius objects}
      \VarNoType{$d$}{$\bot$}
    \end{svariables}

    \begin{procedure}{\propose{u}}
      \While{\true} 
        \If{$d \neq \bot$} \labline{rconsensus:1}
          \Return $d$ \labline{rconsensus:2}
        \EndIf
        \State $o \assign R.\enter{}$ \labline{rconsensus:3}
        \State $(f,u) \assign o.\adoptcommit{u}$ \labline{rconsensus:4}
        \If{$f = \flagCommit$} \labline{rconsensus:5}
          \State $d \assign u$ \labline{rconsensus:6}
        \EndIf
      \EndWhile
    \end{procedure}

  \end{algorithmic}

\end{algorithm}



%% file: algorithms/runiv.tex
\begin{algorithm}[t]

  \caption{Universal Construction -- code at process $p$}
  \labalg{runiv} 
  \footnotesize

  \begin{algorithmic}[1]

    \begin{svariables}
      \VarCustom{$R$} \Comment{A racing on consensus objects}
    \end{svariables}
    
    \begin{lvariables}
      \VarNoType{$C$}{$R.\enter{}$}
      \VarNoType{s}{\stateInit}
    \end{lvariables}

    \begin{procedure}{$\invoke{\op}$}
      \While{\true} \labline{runiv:1}
        \State $d \assign C.d$ \labline{runiv:2}
        \If{$d \neq \bot$} \labline{runiv:3}
          \State $s \assign d[1]$ \labline{runiv:4}
          \State $C \assign R.\enter{}$ \labline{runiv:5} 
        \Else 
          \State $(s',v) \assign \tau(s,\op)$ \labline{runiv:6} 
          \If{$s \equals s'$} \labline{runiv:7}
            \Return $v$ \labline{runiv:8}
          \EndIf
          \State $d \assign C.\propose{(p,s')}$ \labline{runiv:9}
          \If{$d[0] = p$} \labline{runiv:10}
            \Return $v$ \labline{runiv:11}
          \EndIf
        \EndIf
      \EndWhile
    \end{procedure}
    
  \end{algorithmic}

\end{algorithm}

%% file: recycle.tex
\section{Managing Memory Usage}
\labsection{recycle}

Every time the state of the object implemented by the universal construction changes, \refalg{runiv} accesses a new consensus instance.
This implies that the number of consensus instances is not bounded and may rapidly exhaust available memory.
In this section, we present a novel recycling technique that addresses this problem.
To that end, we first introduce several definitions that capture the notion of recycled objects.
Then, we depict a mechanism to recycle the objects used in \refalg{runiv}.

\subsection{Preliminary Notions}
\labsection{recycle:preliminary}

Intuitively, every time an object is reused, it should behave according to its specification.
We formalize this idea in the definitions that follow.

\begin{definition}[\emph{Round} \& \emph{Decomposition}]
  \labdef{roundsanddecomposition}
  Given some history $h$, a round $r$ of $h$ is a sub-history of $h$ such that every invocation complete in $h$ is complete in $r$.
  A \emph{decomposition} of $h$ is an ordered set of rounds $\{r_{1}. \ldots.r_{m \geq 1}\}$ satisfying $h=r_{1}.\ldots.r_{m}$.
\end{definition}

\begin{definition}[\emph{Recycled Object}]
  \labdef{recycled}
  Consider a history $h$ of some object $o$.
  We say that $o$ is a \emph{recycled} object of type $\mathfrak{T}$ during $h$,
  when there exists a decomposition of $h$ such that every round $r$ in this decomposition is a correct history for an object of type $\mathfrak{T}$.
\end{definition}

In order to illustrate these definitions, let us consider two processes $p$ and $q$, and a shared object $o$ exporting an operation $\op$.
We can decompose history
$h_1=\invocation{p,1}{\op}.\invocation{q,1}{\op}.\response{q,1}{\op}{u}.\response{p,1}{\op}{v}.\invocation{p,2}{\op}$
in rounds $r_1=\invocation{p,1}{\op}.\invocation{q,1}{\op}.\response{q,1}{\op}{u}.\response{p,1}{\op}{v}$ and $r_2=\invocation{p,2}{\op}$.
However, if we consider that $\op=\proposeAction$ and $u \neq v$, 
there is no decomposition of $h_1$ for which $o$ is a recycled consensus object.

The usual approach to recycle an object is to reset all its fields once the processes have stopped accessing it,
that is once all the operations pending in a round have completed.
The universal construction of \citet{waitfree} implements this idea 
by provisioning for each process $O(n^2)$ cells, each cell storing the state of the implemented object.
An array of $O(n)$ bits associated to each cell indicates when it can be reset by its owner.

Since the participants to the universal construction are unknown in our context, we cannot employ the previous approach.
Instead, we propose to recycle the objects used in \refalg{runiv} by signing each modification with the round at which it occurs.
An operation that updates such an object will be oblivious to modifications made in prior rounds.
If now the operation is in late, that is when a new round has started before it returns,
the operation will observe the object in a state consistent with one of the rounds to which it is concurrent.
We develop this idea in the next section, then apply it to \refalg{runiv}.

\subsection{Recycling Objects}
\labsection{recycle:recycling}

As a starter, let us remind the definition of a decidable object.
This category of objects contains consensus, but also the splitter and the grafarius objects we described in \refsection{base}.

\begin{definition}[Decidable Object]
  A decidable object $o$ is a shared object whose state contains a \emph{decision} register $d$ 
  taking its value in some set $\mathcal{V}$, the domain of $o$, union $\bot \notin \mathcal{V}$, and which initially equals $\bot$.
  The object is said \emph{decided} when $d \in \mathcal{V}$ holds.
  For every operation $\op$ of $o$, once $o$ is decided, 
  there exists some deterministic function $f$ of $d$ such that $f(d)$ is a valid response value for $\op$.
\end{definition}

As an example of the previous definition, let us consider a grafarius object.
We observe that when the decision register $d$ does not equal $\bot$,
the pair $(\flagAdopt,d)$ is a sound response for the call \adoptcommitOperation.

The first step of our recycling mechanism consists in recycling the MWMR registers that form the basic building blocks of our algorithms.
We detail it below.
\begin{construction}
  Let $(\timestampSet,<)$ be a set of timestamps totally ordered by some relation $<$
  and containing a smallest element $0 \in \timestampSet$.
  For every register $x$ having some initial state $\stateInit$, we initialize $x$ to $(0,\stateInit)$.
  Then, consider some timestamp $t$.
  When a value $v$ is written to $x$, we write $(t,v)$ to $x$.
  Now, upon reading from $x$,
  the value returned is the value $u$ in case $x$ contains $(t',u)$ with $t \leq t'$,  and $\stateInit$ otherwise.
\end{construction}

\noindent
In a second step, we extend this technique to decidable objects as follows.
\begin{construction}
  For some decidable object $o$, a call to $\recycle{o,t}$ returns a copy of $o$ such that upon a call to an operation $\op$ of $o$,%
  \begin{inparaenum}[(i)]
  \item if the object is decided then we return $f(d)$, and otherwise
  \item $\op$ is executed but read and write operations 
    on the shared registers that implement $o$ are replaced by the above construction using timestamp $t$.
  \end{inparaenum}
\end{construction}
For some decidable object $o$, 
we shall write $\recycle{o}$ the object obtained by proxying every call to the operations of $o$ 
by corresponding calls to $\recycle{o,t}$ for some timestamp $t$.
\refprop{recycle} proves that, provided the timestamps are appropriate, $\recycle{o}$ implements a recycled object of the same type as $o$.

\begin{proposition}
  \labprop{recycle}
  Consider a decidable object $o$ of type $\mathfrak{T}$ and some history $h$ of $\recycle{o}$ during which the following invariant holds:
  \begin{compactitem}
  \item[(P1)] %
    For any pair of operations $\op$ and $\op'$, executed respectively on $\recycle{o,t}$ and $\recycle{o,t'}$ in $h$, 
    if $\op'$ does not precede $\op$ in $h$ and $t'<t$ holds, 
    there exists an operation on $\recycle{o,t'}$ that precedes $\op'$ in $h$ and writes to the decision register $d$ of $o$.
  \end{compactitem}
  Then, $\recycle{o}$ implements a recycled object of type $\mathfrak{T}$ during history $h$.
\end{proposition}

\input{algorithms/buniv.tex}

\subsection{Application}
\labsection{recycle:univ}

\refalg{buniv} depicts our second obstruction-free universal construction.
In comparison to \refalg{runiv}, we introduce two modifications:%
\begin{inparaenum}[(i)]
\item processes now compete to decide which consensus will store the next state of the object, and
\item consensus objects are recycled using the mechanism we presented in Construction~2.
\end{inparaenum}

With more details, \refalg{buniv} works as follows.
We implement a racing on consensus objects with variables $L$ and $F$.
When an operation changes the state of the implemented object,  
the calling process proposes to consensus the new state $s'$ together with the index of the consensus object 
that will be used next and its associated timestamp (\refline{buniv:16}).
The index is determined by a call to the function \free.
This function retrieves the codomain of $L$, 
and computes the smallest consensus index that is not currently accessed by a process (\reflines{buniv:1}{buniv:5}). 
In case all objects between $0$ and $\Gamma$ are busy, where $\Gamma$ is the greatest index accessed so far, the index $\Gamma+1$ is returned.

\refalg{buniv} recycles the consensus objects in the codomain of $\ifunct$ using the timestamping schema we introduced in \refsection{recycle:recycling}.
During an execution, for every object \recycle{o} with $o \in \codomainOf{\ifunct}$, the algorithm maintains the invariant P1 of \refprop{recycle}.
This ensures that accesses to variables $L$ and $F$ implement a racing on consensus objects, reducing \refalg{buniv} to \refalg{runiv}.

\textbf{Time \& Space Complexity.}
The contention-free time complexity of \refalg{buniv} is the same as \refalg{runiv}, i.e., it belongs to $\Theta(k)$ in \asm.
From the code of function \free, \refalg{buniv} employs at most $k+1$ consensus objects.
In a distributed system, a quorum system can implement a collect object by emulating $O(k)$ shared registers.
It results that in such a model the contention-free time complexity of \refalg{buniv} measured in message delay is $O(1)$, 
and that its space complexity belongs to $O(k^2)$.

%% file: algorithms/buniv.tex
\begin{algorithm}[t]

  \caption{Universal Construction -- code at process $p$}
  \labalg{buniv} 
  \footnotesize

  \begin{algorithmic}[1]

    \begin{svariables}
      \VarCustom{$L$} \Comment{A map from $\procSet$ to $\naturalSet$, initially \emptySet} \labline{buniv:var:1}
    \end{svariables}
    
    \begin{lvariables}
      \VarCustom{$\ifunct$} \Comment{A function from $\naturalSet$ to consensus objects}
      \VarNoType{$s$}{$\stateInit$}
      \VarNoType{$\lastlap$}{$0$}
      \VarNoType{$\timestamp$}{$0$}
      \VarNoType{$C$}{$\enter{}$}
    \end{lvariables}

    \vspace{-1.5em} 
    \begin{multicols}{2}%
    \begin{procedure}{\invoke{\op}}
      \While{\true} \labline{buniv:8}
        \State $d \assign C.d$ \labline{buniv:9}
        \If{$d \neq \bot$} \labline{buniv:10}
          \State $(s, \lastlap, \timestamp) \assign (d[1],d[2],d[3])$ \labline{buniv:11} 
          \State $C \assign \enter{}$ \labline{buniv:12} 
        \Else 
          \State $(s',v) \assign \tau(s,\op)$ \labline{buniv:13} 
          \If{$s \equals s'$} \labline{buniv:14}
            \Return $v$ \labline{buniv:15}
          \EndIf
          \State $d \assign C.\propose{(p,s',\free,\timestamp+1)}$ \labline{buniv:16}
          \If{$d[0]=p$} \labline{buniv:17}
            \Return $v$ \labline{buniv:18}
          \EndIf
        \EndIf
      \EndWhile
    \end{procedure}
    \columnbreak

    \begin{function}{$\free$}
        \State $S \assign \codomainOf{L}$ \labline{buniv:1} 
        \State \textbf{let} $(\Gamma, \gamma) = (\maxOf{S},\minOf{S})$ \labline{buniv:2} 
        \If{$\gamma > 0$} \labline{buniv:3} 
          \Return $\gamma-1$ \labline{buniv:4} 
        \EndIf
        \Return $\Gamma+1$ \labline{buniv:5} 
    \end{function}
    \vspace{1em}

    \begin{function}{$\enter{}$}
      \State $L[p] \assign \lastlap$ \labline{buniv:6} 
      \Return $\recycle{\ifunctOf{\lastlap},\timestamp}$ \labline{buniv:7}
    \end{function}
  \end{multicols}
  \vspace{-\baselineskip}
  \end{algorithmic}

\end{algorithm}

%% file: empirical.tex
\section{Empirical Assessment}
\labsection{empirical}

To assess the practicability of our approach, we evaluate in this section a prototype implementation of \refalg{buniv}.
This implementation is built on top of the Apache Cassandra distributed data store~\cite{cassandra} which 
provides a distributed shared memory using consistent hashing and quorums of configurable sizes.
In what follows, we describe the internals of our implementation then detail its performance in comparison 
to the Apache Zookeeper coordination service~\cite{zookeeper}.
For the sake of reproducibility, the source code of our implementation as well as the scripts 
we run during the experiments are publicly available~\cite{pssolib}.

\subsection{Implementation Details}
\labsection{empirical:proto}

\paragraph{Cassandra} offers a data model close to the classical relational model at core of the database systems. 
The smallest data unit in Cassandra is a column, a tuple that contains a name, a value and a timestamp.
Columns are grouped by rows, and a column family contains a set of rows.
Each row is indexed by a key, and stored at a quorum of replicas (following a consistent hashing strategy).
A client can read a whole row and write a column.
The consistency of such operations is tunable in Cassandra.
When the cluster running Cassandra is synchronized and both read and write operations operate on quorums, Cassandra provides an atomic snapshot model.
This storage system also supports eventually consistent operations.
When this consistency level is employed, a write operation accesses a quorum of replicas, while a read occurs at a single replica.
Cassandra reconciles replicas via a timestamp-based mechanism in the background.

\vspace{-0.5em}
\paragraph{Prototype Implementation}
Our implementation uses the Python programming language 
and it consists of the different shared objects we detailed in the previous sections (splitter, grafarius, consensus, and universal construction).
The conciseness of Python allows the whole implementation to use around 1,000 lines of code.
Our implementation closely follows the pseudo-code of the algorithms.
Each object corresponds to a row in a column family, and is named after the type of the object.
When an object relies on lower-level abstractions, e.g., consensus employs multiple grafarius objects,
the objects' keys at the low-level are named after the key at the higher one, e.g., \emph{consensus:12:grafarius:3}.
By changing the consistency of the decision register $d$ in \refalg{rconsensus}, we can tune the consistency of the universal abstraction.
When $d$ is eventually consistent, the universal abstraction is sequentially consistent for trivial operations; otherwise it is linearizable.
In Zookeeper, updates are linearizable while read operations are sequentially consistent.
For that reason, when we compare the performance of our implementation to Zookeeper during the experiments, 
we use the sequentially consistent variation of our algorithm.

\subsection{Evaluation}
\labsection{empirical:eval}

We conducted our experiments on a cluster of virtualized 8-core Xeon 2.5~Ghz machines running Gentoo Linux, and connected with a virtualized 1~Gbps switched network.
Network performance, as measured by ping and netperf, are 0.3~ms for a round-trip and a bandwidth of 117MB/s.
Each machine is equipped with a virtual hard-drive whose read/write (uncached) performance, as measured with \texttt{hdparm} and \texttt{dd}, are 246/200~MB/s.
A \emph{server machine} runs either Cassandra or Zookeeper.
A \emph{client machine} emulates multiple clients accessing concurrently the shared objects.
During an experiment, a client executes $10^4$ accesses to one or more objects.
We used 1 to 20 clients machines, emulating 1 to 100 clients each, and 3 to 12 server machines.
In all our experiments, the client machines were not a bottleneck.

\begin{wrapfigure}{r}{0.4\textwidth}
  \vspace{-2.9em}
  \centering
  \scriptsize
  \input{results/cas/same/same.tex} 
  \vspace{-2.3em}
  \caption{\labfigure{cas}Compare-And-Swap}
  \vspace{-0.1em}
\end{wrapfigure}
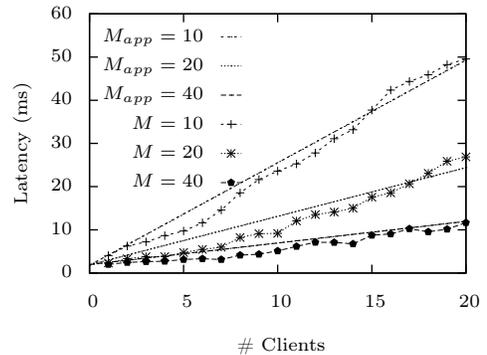

\vspace{-0.5em}
\paragraph{Compare-and-swap}
We first evaluate in \reffigure{cas} the performance of our implementation 
when clients execute compare-and-swap operations, and the system is composed of 3 server machines.
Recall that a compare-and-swap object exposes a single operation: $\compareandswap{u,v}$.
This operation ensures that if the old value of the object equals $u$, it is replaced by $v$.
In such a case, the operation returns \true; otherwise it returns \false.
In \reffigure{cas}, we plot the latency to execute a compare-and-swap operation
as a function of the number of clients and the arguments of the operations.
The initial state of the compare-and-swap object is $0$.
Each client executes in closed-loop an operation \compareandswap{k,l}
where $k$ and $l$ are taken uniformly at random from the interval $\llbracket 0,M \llbracket$ and $M$ is some maximum value. 

When the size of the interval $\llbracket 0,M \llbracket$ shrinks, 
each \compareandswap{} operation has more chance to success in transforming the state of the object, and thus contention increases.
Consequently, as observed in \reffigure{cas}, performance degrades.
The contention between clients occur mainly on the splitter objects that form the building blocks of \refalg{grafarius}.
We briefly analyze how contention is related to performance below.

A compare-and-swap operation $\compareandswap{u,v}$ is successful when the state of the object is $u$; otherwise it fails.
Let us note $\rho$ the ratio of successful operations, that is $1/M$, 
and $\lambda_s$ (respectively $\lambda_f$) the latency to execute solo a successful (resp. failed) operation.
In \reffigure{cas}, the light lines ($M_{approx}$) plot for each value of $M$ the curve $\lambda_f  (1-\rho) + \lambda_s  \rho  n$.
This is a reasonable approximation where the term $\lambda_s \rho n$ follows Little's law~\cite{Allen} and translates the convoy effect~\cite{Blasgen} on successful operations.

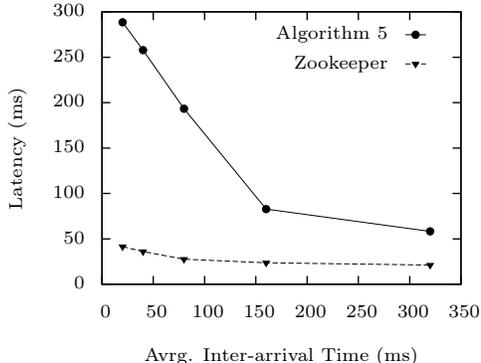
\begin{wrapfigure}{r}{0.4\textwidth}
  \vspace{-3.5em}
  \centering
  \scriptsize
  \input{results/lock/same.tex}
  \vspace{-1.9em}
  \caption{\labfigure{cs}Critical Section}
  \vspace{-0.5em}
\end{wrapfigure}

\vspace{-0.5em}
\paragraph{Critical section}
In \reffigure{cs}, we compare the performance of our implementation to Zookeeper, when clients access a critical section (CS).
Such an object is not in line with the non-blocking approach, but it is commonly used in distributed applications.
We implemented the CS on top of our universal construction using a back-off mechanism.
For Zookeeper, we employed the recipe described in~\cite{zookeeper}.
\reffigure{cs} presents the average time a client takes to enter then leave the CS,
and we vary the inter-arrival time of clients in the critical section according to a Poisson distribution.

We observe in \reffigure{cs} that when the inter-arrival time is high, and thus little contention occurs,
a client accesses the CS with Zookeeper in 20~ms.
For \refalg{buniv}, it takes 60~ms, but the performance degrades quickly when clients access more frequently the CS.
This comes from the fact that%
\begin{inparaenum}[(i)]
\item we implemented a spinlock and thus clients are constantly accessing the system, and
\item as pointed out previously, when clients are competing on splitter objects, the performance of our algorithm degrades.
\end{inparaenum}

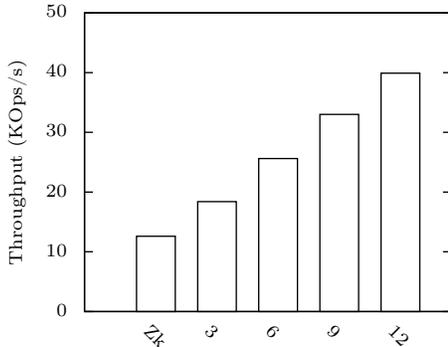
\begin{wrapfigure}{r}{0.4\textwidth}
  \vspace{-2.5em}
  \centering
  \scriptsize
  \input{results/cas/different/different.tex}
  \vspace{-2.7em}
  \caption{\labfigure{scal}Scalability}
  \vspace{-1em}
\end{wrapfigure}

\vspace{-0.5em}
\paragraph{Scalability}
In our last experiment, we assess the scalability of our approach.
To that end, we compute the maximal throughput of our prototype implementation 
when clients access different objects, precisely \compareandswap{} for $M=10$. 
The amount of available server machines varies from 3 to 12 servers;
in all cases we implement a register with the help of a quorum of 3 servers.
We compare our results to an instance of 3 Zookeeper machines.
Zookeeper does not implement natively a compare-and-swap operation.
We devised the following implementation relying on the versioning mechanism exposed to the clients by Zookeeper.
When a client executes \compareandswap{u,v} 
it first retrieves the value $w$ and the attached version $k$ of the znode identifying the object.
In case $w=u$, the client writes $v$ with version $k+1$.
If this write fails due to a concurrent writing, the client re-execute \compareandswap{u,v}.
Note that in our experiments since a single client accesses each object, 
it never retries an operation and this implementation is therefore optimal.

\reffigure{scal} depicts the results obtained for 3 to 12 servers.
With 3 servers, our system delivers 18.4~KOps/s and ZooKeeper 12.6~KOps/s.
This gap is explained by the bottleneck nature of the leader in ZooKeeper which serializes all updates.
Our prototype performance increases to 33K when 9 servers are used, and to 40K with 12 servers.
In this last case, our system is 3.2 times faster than a Zookeeper instance on 3 machines.

%% file: results/cas/same/same.tex
\begingroup
  \makeatletter
  \providecommand\color[2][]{%
    \GenericError{(gnuplot) \space\space\space\@spaces}{%
      Package color not loaded in conjunction with
      terminal option `colourtext'%
    }{See the gnuplot documentation for explanation.%
    }{Either use 'blacktext' in gnuplot or load the package
      color.sty in LaTeX.}%
    \renewcommand\color[2][]{}%
  }%
  \providecommand\includegraphics[2][]{%
    \GenericError{(gnuplot) \space\space\space\@spaces}{%
      Package graphicx or graphics not loaded%
    }{See the gnuplot documentation for explanation.%
    }{The gnuplot epslatex terminal needs graphicx.sty or graphics.sty.}%
    \renewcommand\includegraphics[2][]{}%
  }%
  \providecommand\rotatebox[2]{#2}%
  \@ifundefined{ifGPcolor}{%
    \newif\ifGPcolor
    \GPcolorfalse
  }{}%
  \@ifundefined{ifGPblacktext}{%
    \newif\ifGPblacktext
    \GPblacktexttrue
  }{}%
  \let\gplgaddtomacro\g@addto@macro
  \gdef\gplbacktext{}%
  \gdef\gplfronttext{}%
  \makeatother
  \ifGPblacktext
    \def\colorrgb#1{}%
    \def\colorgray#1{}%
  \else
    \ifGPcolor
      \def\colorrgb#1{\color[rgb]{#1}}%
      \def\colorgray#1{\color[gray]{#1}}%
      \expandafter\def\csname LTw\endcsname{\color{white}}%
      \expandafter\def\csname LTb\endcsname{\color{black}}%
      \expandafter\def\csname LTa\endcsname{\color{black}}%
      \expandafter\def\csname LT0\endcsname{\color[rgb]{1,0,0}}%
      \expandafter\def\csname LT1\endcsname{\color[rgb]{0,1,0}}%
      \expandafter\def\csname LT2\endcsname{\color[rgb]{0,0,1}}%
      \expandafter\def\csname LT3\endcsname{\color[rgb]{1,0,1}}%
      \expandafter\def\csname LT4\endcsname{\color[rgb]{0,1,1}}%
      \expandafter\def\csname LT5\endcsname{\color[rgb]{1,1,0}}%
      \expandafter\def\csname LT6\endcsname{\color[rgb]{0,0,0}}%
      \expandafter\def\csname LT7\endcsname{\color[rgb]{1,0.3,0}}%
      \expandafter\def\csname LT8\endcsname{\color[rgb]{0.5,0.5,0.5}}%
    \else
      \def\colorrgb#1{\color{black}}%
      \def\colorgray#1{\color[gray]{#1}}%
      \expandafter\def\csname LTw\endcsname{\color{white}}%
      \expandafter\def\csname LTb\endcsname{\color{black}}%
      \expandafter\def\csname LTa\endcsname{\color{black}}%
      \expandafter\def\csname LT0\endcsname{\color{black}}%
      \expandafter\def\csname LT1\endcsname{\color{black}}%
      \expandafter\def\csname LT2\endcsname{\color{black}}%
      \expandafter\def\csname LT3\endcsname{\color{black}}%
      \expandafter\def\csname LT4\endcsname{\color{black}}%
      \expandafter\def\csname LT5\endcsname{\color{black}}%
      \expandafter\def\csname LT6\endcsname{\color{black}}%
      \expandafter\def\csname LT7\endcsname{\color{black}}%
      \expandafter\def\csname LT8\endcsname{\color{black}}%
    \fi
  \fi
  \setlength{\unitlength}{0.0500bp}%
  \begin{picture}(3960.00,2923.20)%
    \gplgaddtomacro\gplbacktext{%
      \csname LTb\endcsname%
      \put(594,704){\makebox(0,0)[r]{\strut{} 0}}%
      \put(594,1030){\makebox(0,0)[r]{\strut{} 10}}%
      \put(594,1356){\makebox(0,0)[r]{\strut{} 20}}%
      \put(594,1682){\makebox(0,0)[r]{\strut{} 30}}%
      \put(594,2007){\makebox(0,0)[r]{\strut{} 40}}%
      \put(594,2333){\makebox(0,0)[r]{\strut{} 50}}%
      \put(594,2659){\makebox(0,0)[r]{\strut{} 60}}%
      \put(726,484){\makebox(0,0){\strut{} 0}}%
      \put(1435,484){\makebox(0,0){\strut{} 5}}%
      \put(2144,484){\makebox(0,0){\strut{} 10}}%
      \put(2854,484){\makebox(0,0){\strut{} 15}}%
      \put(3563,484){\makebox(0,0){\strut{} 20}}%
      \put(220,1681){\rotatebox{-270}{\makebox(0,0){\scriptsize Latency (ms)}}}%
      \put(2144,154){\makebox(0,0){\strut{} {\scriptsize \# Clients}}}%
    }%
    \gplgaddtomacro\gplfronttext{%
      \csname LTb\endcsname%
      \put(1584,2486){\makebox(0,0)[r]{\strut{}$M_{app}=10$}}%
      \csname LTb\endcsname%
      \put(1584,2266){\makebox(0,0)[r]{\strut{}$M_{app}=20$}}%
      \csname LTb\endcsname%
      \put(1584,2046){\makebox(0,0)[r]{\strut{}$M_{app}=40$}}%
      \csname LTb\endcsname%
      \put(1584,1826){\makebox(0,0)[r]{\strut{}$M=10$}}%
      \csname LTb\endcsname%
      \put(1584,1606){\makebox(0,0)[r]{\strut{}$M=20$}}%
      \csname LTb\endcsname%
      \put(1584,1386){\makebox(0,0)[r]{\strut{}$M=40$}}%
    }%
    \gplbacktext
    \put(0,0){\includegraphics{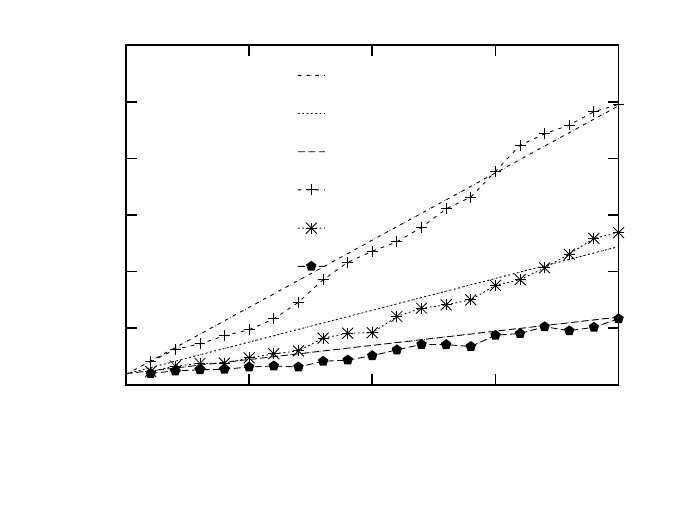}}%
    \gplfronttext
  \end{picture}%
\endgroup

%% file: results/lock/same.tex
\begingroup
  \makeatletter
  \providecommand\color[2][]{%
    \GenericError{(gnuplot) \space\space\space\@spaces}{%
      Package color not loaded in conjunction with
      terminal option `colourtext'%
    }{See the gnuplot documentation for explanation.%
    }{Either use 'blacktext' in gnuplot or load the package
      color.sty in LaTeX.}%
    \renewcommand\color[2][]{}%
  }%
  \providecommand\includegraphics[2][]{%
    \GenericError{(gnuplot) \space\space\space\@spaces}{%
      Package graphicx or graphics not loaded%
    }{See the gnuplot documentation for explanation.%
    }{The gnuplot epslatex terminal needs graphicx.sty or graphics.sty.}%
    \renewcommand\includegraphics[2][]{}%
  }%
  \providecommand\rotatebox[2]{#2}%
  \@ifundefined{ifGPcolor}{%
    \newif\ifGPcolor
    \GPcolorfalse
  }{}%
  \@ifundefined{ifGPblacktext}{%
    \newif\ifGPblacktext
    \GPblacktexttrue
  }{}%
  \let\gplgaddtomacro\g@addto@macro
  \gdef\gplbacktext{}%
  \gdef\gplfronttext{}%
  \makeatother
  \ifGPblacktext
    \def\colorrgb#1{}%
    \def\colorgray#1{}%
  \else
    \ifGPcolor
      \def\colorrgb#1{\color[rgb]{#1}}%
      \def\colorgray#1{\color[gray]{#1}}%
      \expandafter\def\csname LTw\endcsname{\color{white}}%
      \expandafter\def\csname LTb\endcsname{\color{black}}%
      \expandafter\def\csname LTa\endcsname{\color{black}}%
      \expandafter\def\csname LT0\endcsname{\color[rgb]{1,0,0}}%
      \expandafter\def\csname LT1\endcsname{\color[rgb]{0,1,0}}%
      \expandafter\def\csname LT2\endcsname{\color[rgb]{0,0,1}}%
      \expandafter\def\csname LT3\endcsname{\color[rgb]{1,0,1}}%
      \expandafter\def\csname LT4\endcsname{\color[rgb]{0,1,1}}%
      \expandafter\def\csname LT5\endcsname{\color[rgb]{1,1,0}}%
      \expandafter\def\csname LT6\endcsname{\color[rgb]{0,0,0}}%
      \expandafter\def\csname LT7\endcsname{\color[rgb]{1,0.3,0}}%
      \expandafter\def\csname LT8\endcsname{\color[rgb]{0.5,0.5,0.5}}%
    \else
      \def\colorrgb#1{\color{black}}%
      \def\colorgray#1{\color[gray]{#1}}%
      \expandafter\def\csname LTw\endcsname{\color{white}}%
      \expandafter\def\csname LTb\endcsname{\color{black}}%
      \expandafter\def\csname LTa\endcsname{\color{black}}%
      \expandafter\def\csname LT0\endcsname{\color{black}}%
      \expandafter\def\csname LT1\endcsname{\color{black}}%
      \expandafter\def\csname LT2\endcsname{\color{black}}%
      \expandafter\def\csname LT3\endcsname{\color{black}}%
      \expandafter\def\csname LT4\endcsname{\color{black}}%
      \expandafter\def\csname LT5\endcsname{\color{black}}%
      \expandafter\def\csname LT6\endcsname{\color{black}}%
      \expandafter\def\csname LT7\endcsname{\color{black}}%
      \expandafter\def\csname LT8\endcsname{\color{black}}%
    \fi
  \fi
  \setlength{\unitlength}{0.0500bp}%
  \begin{picture}(3960.00,3024.00)%
    \gplgaddtomacro\gplbacktext{%
      \csname LTb\endcsname%
      \put(726,704){\makebox(0,0)[r]{\strut{} 0}}%
      \put(726,1046){\makebox(0,0)[r]{\strut{} 50}}%
      \put(726,1389){\makebox(0,0)[r]{\strut{} 100}}%
      \put(726,1731){\makebox(0,0)[r]{\strut{} 150}}%
      \put(726,2074){\makebox(0,0)[r]{\strut{} 200}}%
      \put(726,2416){\makebox(0,0)[r]{\strut{} 250}}%
      \put(726,2759){\makebox(0,0)[r]{\strut{} 300}}%
      \put(858,484){\makebox(0,0){\strut{} 0}}%
      \put(1244,484){\makebox(0,0){\strut{} 50}}%
      \put(1631,484){\makebox(0,0){\strut{} 100}}%
      \put(2017,484){\makebox(0,0){\strut{} 150}}%
      \put(2404,484){\makebox(0,0){\strut{} 200}}%
      \put(2790,484){\makebox(0,0){\strut{} 250}}%
      \put(3177,484){\makebox(0,0){\strut{} 300}}%
      \put(3563,484){\makebox(0,0){\strut{} 350}}%
      \put(220,1731){\rotatebox{-270}{\makebox(0,0){Latency (ms)}}}%
      \put(2210,154){\makebox(0,0){\strut{}Avrg. Inter-arrival Time (ms)}}%
    }%
    \gplgaddtomacro\gplfronttext{%
      \csname LTb\endcsname%
      \put(3012,2586){\makebox(0,0)[r]{\strut{}\refalg{buniv}}}%
      \csname LTb\endcsname%
      \put(3012,2366){\makebox(0,0)[r]{\strut{}Zookeeper}}%
    }%
    \gplbacktext
    \put(0,0){\includegraphics{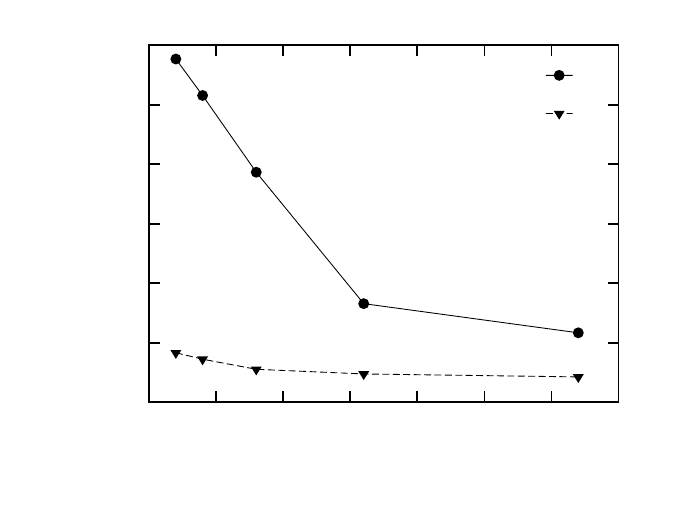}}%
    \gplfronttext
  \end{picture}%
\endgroup

%% file: results/cas/different/different.tex
\begingroup
  \makeatletter
  \providecommand\color[2][]{%
    \GenericError{(gnuplot) \space\space\space\@spaces}{%
      Package color not loaded in conjunction with
      terminal option `colourtext'%
    }{See the gnuplot documentation for explanation.%
    }{Either use 'blacktext' in gnuplot or load the package
      color.sty in LaTeX.}%
    \renewcommand\color[2][]{}%
  }%
  \providecommand\includegraphics[2][]{%
    \GenericError{(gnuplot) \space\space\space\@spaces}{%
      Package graphicx or graphics not loaded%
    }{See the gnuplot documentation for explanation.%
    }{The gnuplot epslatex terminal needs graphicx.sty or graphics.sty.}%
    \renewcommand\includegraphics[2][]{}%
  }%
  \providecommand\rotatebox[2]{#2}%
  \@ifundefined{ifGPcolor}{%
    \newif\ifGPcolor
    \GPcolortrue
  }{}%
  \@ifundefined{ifGPblacktext}{%
    \newif\ifGPblacktext
    \GPblacktexttrue
  }{}%
  \let\gplgaddtomacro\g@addto@macro
  \gdef\gplbacktext{}%
  \gdef\gplfronttext{}%
  \makeatother
  \ifGPblacktext
    \def\colorrgb#1{}%
    \def\colorgray#1{}%
  \else
    \ifGPcolor
      \def\colorrgb#1{\color[rgb]{#1}}%
      \def\colorgray#1{\color[gray]{#1}}%
      \expandafter\def\csname LTw\endcsname{\color{white}}%
      \expandafter\def\csname LTb\endcsname{\color{black}}%
      \expandafter\def\csname LTa\endcsname{\color{black}}%
      \expandafter\def\csname LT0\endcsname{\color[rgb]{1,0,0}}%
      \expandafter\def\csname LT1\endcsname{\color[rgb]{0,1,0}}%
      \expandafter\def\csname LT2\endcsname{\color[rgb]{0,0,1}}%
      \expandafter\def\csname LT3\endcsname{\color[rgb]{1,0,1}}%
      \expandafter\def\csname LT4\endcsname{\color[rgb]{0,1,1}}%
      \expandafter\def\csname LT5\endcsname{\color[rgb]{1,1,0}}%
      \expandafter\def\csname LT6\endcsname{\color[rgb]{0,0,0}}%
      \expandafter\def\csname LT7\endcsname{\color[rgb]{1,0.3,0}}%
      \expandafter\def\csname LT8\endcsname{\color[rgb]{0.5,0.5,0.5}}%
    \else
      \def\colorrgb#1{\color{black}}%
      \def\colorgray#1{\color[gray]{#1}}%
      \expandafter\def\csname LTw\endcsname{\color{white}}%
      \expandafter\def\csname LTb\endcsname{\color{black}}%
      \expandafter\def\csname LTa\endcsname{\color{black}}%
      \expandafter\def\csname LT0\endcsname{\color{black}}%
      \expandafter\def\csname LT1\endcsname{\color{black}}%
      \expandafter\def\csname LT2\endcsname{\color{black}}%
      \expandafter\def\csname LT3\endcsname{\color{black}}%
      \expandafter\def\csname LT4\endcsname{\color{black}}%
      \expandafter\def\csname LT5\endcsname{\color{black}}%
      \expandafter\def\csname LT6\endcsname{\color{black}}%
      \expandafter\def\csname LT7\endcsname{\color{black}}%
      \expandafter\def\csname LT8\endcsname{\color{black}}%
    \fi
  \fi
  \setlength{\unitlength}{0.0500bp}%
  \begin{picture}(3888.00,2923.20)%
    \gplgaddtomacro\gplbacktext{%
      \csname LTb\endcsname%
      \put(594,406){\makebox(0,0)[r]{\strut{} 0}}%
      \put(594,857){\makebox(0,0)[r]{\strut{} 10}}%
      \put(594,1307){\makebox(0,0)[r]{\strut{} 20}}%
      \put(594,1758){\makebox(0,0)[r]{\strut{} 30}}%
      \put(594,2208){\makebox(0,0)[r]{\strut{} 40}}%
      \put(594,2659){\makebox(0,0)[r]{\strut{} 50}}%
      \put(1187,274){\rotatebox{-45}{\makebox(0,0)[l]{\strut{}Zk}}}%
      \put(1648,274){\rotatebox{-45}{\makebox(0,0)[l]{\strut{}3}}}%
      \put(2108,274){\rotatebox{-45}{\makebox(0,0)[l]{\strut{}6}}}%
      \put(2569,274){\rotatebox{-45}{\makebox(0,0)[l]{\strut{}9}}}%
      \put(3030,274){\rotatebox{-45}{\makebox(0,0)[l]{\strut{}12}}}%
      \put(220,1532){\rotatebox{-270}{\makebox(0,0){\scriptsize Throughput (KOps/s)}}}%
    }%
    \gplgaddtomacro\gplfronttext{%
    }%
    \gplbacktext
    \put(0,0){\includegraphics{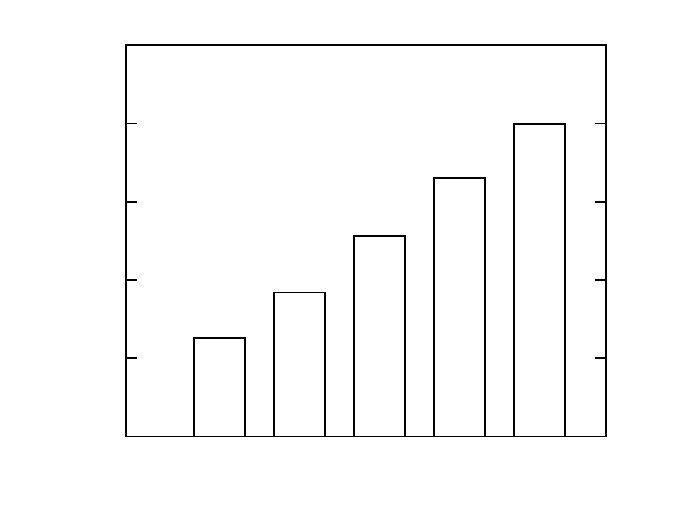}}%
    \gplfronttext
  \end{picture}%
\endgroup

%% file: conclusion.tex
\section{Conclusion}
\labsection{conclusion}

This paper presents a novel algorithmic solution to implement a distributed universal construction when participants are unknown.
Contrary to previous works, which mostly focus on state machine replication, our approach employs solely a distributed asynchronous shared memory,
the logic of consistent operations being delegated to the client side.
Hence, and as exemplified by our prototype, we can implement it in a client library that runs on top of an off-the-shelf distributed shared memory.
To obtain this result, we introduce two new shared abstractions: a grafarius and a racing, which we believe are of interest on their own.
We also present a mechanism to recycle the base objects at core of our construction.

%% file: appendix.tex
\input{appendix/base.tex}
\input{appendix/recycle.tex}
\input{appendix/buniv.tex}

%% file: appendix/base.tex
\section{Correctness of Algorithms \ref{alg:grafarius} to \ref{alg:runiv}}
\labappendix{base}

The two theorems that follow prove the correctness of our implementations of a grafarius and a racing objects.

\begin{theorem}
  \labtheo{grafarius}
  \refalg{grafarius} implements a wait-free grafarius.
\end{theorem}

\begin{proof}
  \refalg{grafarius} employs only wait-free shared objects.
  Hence, to the light of its pseudo-code, it conserves this liveness property.
  Validity and solo convergence are both immediate.
  To prove the continuation property, 
  we observe that when a process $p$ returns from a call to \adoptcommitOperation, 
  \begin{inparaenum}[(i)]
  \item $d \neq \bot$ holds, and 
  \item if a process $q$ invokes \adoptcommitOperation after $p$ returned, it loses the splitter.
  \end{inparaenum}
  It remains to show that \refalg{grafarius} also ensures the coherence property.
  To that end, consider that in some execution $\run$ of \refalg{grafarius}, a process $p$ commits a value $u$.
  Note $t_0$ the time in $\run$ at which $p$ executes \refline{grafarius:8}.
  Since $p$ commits $u$, $c$ always equals $\false$ before $t_0$.
  This implies that $p$ wins the splitter $s$ in $\run$.
  From the properties of a splitter object, 
  we conclude that every process, except $p$, that completes its invocation in $\run$ must execute \refline{grafarius:2}.
  This observation and the fact that $c=\false$ at $t_0$ tell us that no process executes \refline{grafarius:2} before $t_0$.
  Then, observe that $p$ executes \refline{grafarius:7} at some time $t_1<t_0$.
  As a consequence, $d=u$ after $t_0$, and every process other than $p$ that completes 
  its invocation in $\run$ executes \refline{grafarius:4}.
  The coherence property holds in $\run$.
\end{proof}

\begin{theorem}
  \labtheo{uracing}
  \refalg{uracing} implements a wait-free linearizable racing.
\end{theorem}

\begin{proof}
  From the pseudo-code of \refalg{uracing}, we conclude that this algorithm is wait-free.
  Then, consider a complete history $h$ produced by the following mapping:%
  \footnote{
    All our implementations are at least obstruction-free.
    Thus, when proving linearizability, we can simply consider a complete history.
  }
  when $p$ invokes \enter{} in \refalg{uracing}, it calls \enter{p,l} in $h$ with $l=\ifunctOf{\lastlap}$ the response value of the call, 
  and when the call of $p$ returns, we append the corresponding response to $h$.

  Let $(\lin,<_{\lin})$ be the linearization of $h$ induced by the order in which the operations on variable $L$ occurs in $h$.
  In more detail, 
  if \enter{p,l} and \enter{q,l'} are concurrent 
  then we choose the order $\enter{p,l} <_{\lin} \enter{q,l'}$
  if $p$ does not read the value of $L[q]$ written by $q$ (\refline{uracing:1}); 
  otherwise, we choose the converse order.

  Fix $\ll_{\lin}$ as the order between the entered laps in $\lin$ induced by the indexing function $\ifunct$.
  Assume a process $p$ enters some lap $l=\ifunctOf{\lastlap}$ in $\lin$.
  We observe that $p$ executes either \refline{uracing:5} or \refline{uracing:6}.
  In the first case, this implies that $p$ has just left $\ifunctOf{\lastlap-1}$ which 
  is the greatest lap smaller than $l$ in which a process has entered during $\lin$.
  In the second case, we observe that necessarily a process left $l$.
  and moreover that this event occurs in $\lin$ before $p$ enters $l$.
  In both cases, the ordering property holds.
\end{proof}

Below, we prove a key proposition characterizing algorithms that employ a racing on decidable objects.
(The notion of decidable object is recalled in \refsection{recycle:preliminary}.)

\begin{proposition}
  \labprop{racing}
  Consider an algorithm $\alg$ that accesses a racing on decidable objects.
  Suppose that in $\alg$, when some process $p$ enters a new object, the last object left by $p$ is decided.
  Then, during every execution $\run$ of $\alg$, at the time $p$ enters a new object $o$, 
  all the objects $o' \ll_{\run} o$ are decided.
\end{proposition}

\begin{proof}
  By induction.
  The property holds clearly for the first object in the order $\ll_{\run}$.
  Then, consider that some process $p$ enters an object $o$ and assume by induction 
  that all the objects prior to $o$ for the order $\ll_{\run}$ are decided.
  By property (i) of the ordering property of a racing object, 
  we can consider without lack of generality that $p$ is the first process to enter $o$ in $\run$.
  Then, by property (ii), process $p$ left previously the greatest object smaller than $l$ for the order $\ll_{\run}$; 
  name this object $o'$.
  Applying the induction hypothesis, we know that for every object $o'' \ll_{\run} o'$, $o''$ is decided at that time.
  Moreover, before $p$ enters $o$, $\alg$ ensures that $o'$ is decided.
  Hence, the property holds for $o$.
\end{proof}

Based on the above proposition, we prove next that 
\refalg{rconsensus} and \refalg{runiv} implement respectively a consensus and a universal abstraction.

\begin{theorem}
  \labtheo{rconsensus}
  \refalg{rconsensus} implements an obstruction-free linearizable consensus.
\end{theorem}

\begin{proof}
  First of all, we observe that a process which starts executing solo necessarily commits the value 
  it proposes into the next grafarius it enters alone.
  This is ensured by the solo convergence property.
  As a consequence, \refalg{rconsensus} is obstruction-free.
  Besides, the validity of consensus follows from the validity property of a grafarius.
  
  It remains to prove that the agreement property is also satisfied.
  To that end, consider a complete history $h$ produced by an execution $\run$ of \refalg{rconsensus} in which some process $p$ returns a value $v$.
  We observe that a grafarius $o$ returned $(\flagCommit,v)$ before $p$ executes \refline{rconsensus:1}.
  Suppose then that some process $q$ adopts, or commits, a value $u$ by accessing the smallest object $o'$ higher than $o$ for the order $\ll_{\run}$.
  By the ordering property of a racing, when $q$ enters $o'$ either%
  \begin{inparaenum}[(i)]
  \item process $q$ just left $o$, and by the coherence property of a grafarius, it must propose $v$ to $o'$, or
  \item a process $q'$ left $o'$ before $q$ enters it, and by a short induction on the continuation property of a grafarius,
    $q$ must return the value $v$ when it accesses $o'$.
  \end{inparaenum}
  Hence, we have $u=v$.
  
  Next, let us consider that two processes $p$ and $q$ return respectively $u$ and $v$ during $h$.
  Name $o$ and $o'$ the two grafarius accessed to decide these values.
  (If the call of either $p$ or $q$ returns at \refline{rconsensus:1}, the situation boils down to this case.)
  If $o=o'$, the coherence property of a grafarius tells us that $u=v$.
  If now $o \neq o'$, we can assume without lack of generality that $o \ll_{\run} o'$ holds and by induction we also obtain that $u=v$.
  Hence, \refalg{rconsensus} maintains the agreement property of consensus during the history $h$.
\end{proof}

\begin{theorem}
  \labtheo{runiv}
  \refalg{runiv} implements an obstruction-free linearizable universal construction.
\end{theorem}

\begin{proof}
  Variable $R$ is a (wait-free) racing on obstruction-free consensus objects.
  Thus, at the light of its pseudo-code, \refalg{runiv} is obstruction-free.

  We prove now that the implementation is linearizable.
  Consider a complete history $h$ produced by some execution $\run$ of \refalg{runiv}, 
  and let $\mathit{op}$ be some operation in $h$ executed by a process $p$.
  Operation $\mathit{op}$ is \emph{trivial} (respectively \emph{regular}) when it returns at \refline{runiv:8} (resp. \refline{runiv:11}).
  The value of variable $C$ on process $p$ at the time the operation returns is the consensus \emph{associated} to $\mathit{op}$.

  Note $\ll_{\run}$ the order induced by the racing object $R$ on the consensus objects associated to at least one operation.
  Let $\lin$ be the history produced by ordering the operations in $h$ according to their associated consensus such that%
  \begin{inparaenum}[(i)]
  \item trivial operations are ordered after the regular ones having the same associated consensus, and
  \item we keep the order between trivial operations having the same associated consensus.
  \end{inparaenum}
  We prove below that $\lin$ is a linearization of $h$.

  First, we show that $<_{h} \subseteq <_{\lin}$.
  Let $\op$ and $\op'$ be two operations of $h$ executed by respectively $p$ and $q$.
  Consider that $\op <_{h} \op'$ holds, and name $c$ and $c'$ their respective associated consensus.
  \begin{inparaenum}
  \item[(Case $c = c'$.)]
    When the two operations have the same associated consensus, 
    we observe that by the agreement property of consensus they cannot be both regular (\refline{runiv:9}).
    Since the consensus $c$ must be decided for $\op$ to execute, the case $\op$ trivial and $\op'$ regular is impossible 
    Hence, we have $\op$ regular or trivial, and $\op'$ trivial.
    By the properties (i) and (ii) of our construction of the history $\lin$, we conclude that $\op <_{\lin} \op'$ holds.
  \item[(Case $c \ll_{\run} c'$.)]
    By construction, we have $\op <_{\lin} \op'$.
  \item [(Case $c' \ll_{\run} c$.)]
    Operation $\op$ is associated to $c$.
    As a consequence, there exists a time $t$ at which $c.d$ equals $\bot$ when the test at \refline{runiv:3} is executed in $\op$.
    Similarly, we note $t'$ the time at which $c'.d$ equals $\bot$ in $\op'$.
    Since $\op$ precedes $\op'$ in $h$, it must be the case that $t<t'$.
    But $c' \ll_{\run} c$ implies by \refprop{racing} that at time $t'$ consensus $c$ is decided; a contradiction.
  \end{inparaenum}
  
  It remains to prove that $\lin$ is a legal sequence.
  To that end, consider an operation $\op$ by a process $p$ in $\lin$.
  Name $v$ its response value, 
  $c$ its associated consensus,
  and $\op'$ the first regular operation that precedes $\op$ in $\lin$.
  In case $\op$ is trivial, by property (i) of our construction, $\op'$ is the regular operation associated to the same consensus as $\op$. 
  Hence, this operation sets the object to a state $s$ such that $\tau(s,\op)=(s,v)$ (see \reflinesthree{runiv:4}{runiv:6}{runiv:9}).
  Assume now that $\op$ is regular.
  When $p$ enters $c$ at \refline{runiv:5}, by the ordering property of a racing object, we can consider two cases:%
  \begin{inparaenum}[(i)]
  \item Process $p$ has left the greatest lap smaller than $l$ for the order $\ll_{\run}$.
    Then, variable $s$ contains the state of the object after $\op'$ and all the operations
    between $\op'$ and $\op$ are trivial.
    Thus, the response $v$ is correct.
  \item The consensus $c$ was left by some process before $p$ enters it.
    In that case, $c$ is decided at that time; a contradiction to the fact that $\op$ is regular.
  \end{inparaenum}
\end{proof}

%% file: appendix/recycle.tex
\section{Proof of \refprop{recycle}}
\labappendix{recycle}

In this section, we show that when invariant P1 holds in some complete history $h$, 
$\recycle{o}$ is a linearizable implementation of a recycled object of type $\mathfrak{T}$.
To prove this statement, we construct from $h$ a history $\lin$ decomposable in rounds, 
each round of $\lin$ being a correct history for some object of type $\mathfrak{T}$.

Our construction of $\lin$ relies on the following distinction between the operations appearing in $h$:
operations that write to the decision register of $o$ are called \emph{modifiers},
whereas operations that read a value in the decision register $d$ of $o$ and immediately return $f(d)$ are named \emph{observers}.

\begin{claim}
  \labclaim{recycle:1}
  For every observer $\op$, there exists a unique modifier $\op'$ such that%
  \begin{inparaenum}[(i)]
  \item $\op$ observes the decision written by $\op'$, 
  \item $\op'$ is either concurrent or prior to $\op$, and 
  \item there is no modifier $\op''$ such that $\op' <_h \op'' <_h op$.
  \end{inparaenum}
  We shall say in the following that $\op'$ is the modifier \emph{associated} to $\op$.
\end{claim}

\begin{proof}
  Consider an operation $\op$ that executes on $\recycle{o,t}$.
  If $\op$ is an observer, then according to Construction~2, the decision register $d$ contains a value different than $\bot$.
  From Construction~1, we deduce that $\op$ reads a value $(t' \geq t,\msgAny)$ from $d$.
  Since the register $d$ is atomic, a unique operation $\op'$ writes this value to $d$.
  This operation is necessarily concurrent, or prior to $\op$, and moreover there must be no modifier $\op''$ such that $\op' <_h \op'' <_h op$.  
\end{proof}

\begin{claim}
  \labclaim{recycle:2}
  For every modifier $\op$ on $\recycle{o,t}$ in $h$,
  if $\op'$ is a modifier on $\recycle{o,t'}$ and $t'<t$ holds then $\op'$ precedes $\op$ in $h$.
\end{claim}

\begin{proof}
  We proceed by induction.
  Consider that the claim holds for all the modifiers prior to some operation $\op$ on $\recycle{o,t}$.
  Then, assume for the sake of contradiction that an operation concurrent, or following, 
  $\op$ writes to the decision register $d$ of $o$ with a timestamp $t'<t$.
  Name $\op'$ the first such operation.
  By P1, there exists an operation $\op''$ on $\recycle{o,t'}$ that precedes $\op'$ in $h$ and that writes to $d$.
  From our induction hypothesis and the choice of $\op'$, 
  every modifier between the response of $\op''$ and the invocation of $\op'$, 
  writes to $d$ with a timestamp $t'' \geq t'$.
  Hence from the code of Construction~2, $\op'$ observes that $o$ is decided when it is executed;
  a contradiction to the fact that this operation is a modifier.
\end{proof}

We detail how to build history $\lin$ in Construction~3.
Further, we show that $\lin$ is a higher-level view of $h$ in the sense of the Definition 5 in~\cite{Lamport86}.
Then, we prove the existence of a decomposition of $\lin$ for which each round is correct for the type $\mathfrak{T}$.

\begin{construction}
  Consider the invocation and response events $e$ that occur in the history $h$.
  In case $e=\invocation{p}{\recycle{o,t}.op}$, we add $\invocation{p}{\recycle{o}.op}$ to $\lin$;
  otherwise $e=\response{p}{\recycle{o,t}.op}{v}$ for some response value $v$, and we add $\response{p}{\recycle{o}.op}{v}$ to $\lin$.
  The order in which we add those events in $\lin$ is%
  \begin{inparaenum}[(i)]
  \item for modifiers, the order in which they write last to the decision register $d$ of $o$,  and 
  \item for observers, the order in which they read $d$.
  \end{inparaenum}
\end{construction}

By the fact that registers are atomic, Construction~3 implies that $\lin$ is sequential.
Moreover, since every operation reads or writes to $d$ and we order them according to their accesses to $d$, 
the following claim is immediate.

\begin{claim}
  \labclaim{recycle:3}
  For any two operations $\op$ and $\op'$ in $h$, if $\op$ precedes $\op'$ in $h$ then $\op$ precedes $\op'$ in $\lin$
\end{claim}

\noindent
In what follows, we prove the correctness of $\lin$.

\begin{claim}
  \labclaim{recycle:4}
  There exits a decomposition of $\lin$ such that 
  in every round $r$ of this decomposition, $r$ is a correct history for an object of type $\mathfrak{T}$.
\end{claim}

\begin{proof}
  For every timestamp $t \in \timestampSet$ that appears in $h$, 
  we define the round $r_t$ as the sub-sequence of modifiers in $\lin$ that occur on $\recycle{o,t}$ in $h$,
  together with the observers (if any) for which one of these operations is a modifier.
  We order rounds according to their associated timestamps.
  Denoting $\{t_1,\ldots,t_m\}$ the ordered set of timestamps that appears in $h$, 
  we prove below that $\lin=r_{t_1}.\ldots.r_{t_n}$ holds.

  Consider an operation $\op$ accessing $\recycle{o,t}$, and that this operation is complete in $h$.
  According to Construction~3, either $\op$ belongs to $r_t$, 
  or by \refclaim{recycle:1}, there exists a modifier associated to $\op$ in some round $r_{t'}$, and $\op$ belongs to that round.
  We observe that in both cases, $\op$ is complete in the round it belongs to.

  Then, let us consider two operations $\op$ and $\op'$ that belong to rounds $r_{t}$ and $r_{t'>t}$.
  In each of the three cases that follow, we prove that $\op$ and $\op'$ do not interleave.
  (Case $\op$ and $\op'$ modifiers.)
  Since $t<t'$, by \refclaim{recycle:2}, $\op$ precedes $\op'$ in $h$.
  By Construction~3, this also holds in $\lin$.
  (Case $\op$ modifier and $\op'$ observer.)
  Note $\op''$ the modifier in round $r_{t'}$ associated to $\op'$.
  This modifier is executed on $\recycle{o,t'}$, whereas $\op$ is executed on $\recycle{o,t<t'}$.
  Applying \refclaim{recycle:2}, we know that $\op''$ precedes $\op$ in $h$, and thus also in $\lin$.
  By Construction~3, we deduce that $\op <_{\lin} \op'' <_{\lin} \op'$ holds.
  (Case $\op$ observer and $\op'$ modifier.)
  This case is symmetric to the previous one and thus omitted.
  (Case $\op$ and $\op'$ observers.)
  Let $\op_1$ and $\op_2$ be the modifiers respectively associated to $\op$ and $\op'$ in rounds $r_{t}$ and $r_{t'}$.
  Since $t<t'$, by \refclaim{recycle:2}, $\op_1$ precedes $\op_2$ in $h$.
  By Construction~3, we deduce that $\op_1 <_{\lin} \op <_{\lin} \op_2 <_{\lin} \op'$ holds.
  
  It remains to show that every round in the above decomposition is a correct history for an object of type $\mathfrak{T}$.
  To achieve this, let us now consider a round $r_t$ and some operation $\op$ in $r_t$.
  If $\op$ is an observer, then it should return the decision set by its associated modifier in $r_t$, and this modifier precedes $\op$.
  Otherwise, by construction all the operations in $r_t$ are executed on $\recycle{o,t}$.
  and by definition of $\lin$, all the modifiers preceding the round $r_t$ were executed on $\recycle{o,t'}$, for some timestamp $t'<t$.
  Thus, from Construction~1,  all the modifiers in $r_t$ are oblivious of the modifications that occur in previous rounds.
  As a consequence, we conclude that $r_t$ is a correct history for an object of type $\mathfrak{T}$.
\end{proof}

%% file: appendix/buniv.tex
\section{Correctness of \refalg{buniv}}
\labappendix{buniv}

As pointed out in \refsection{recycle:univ}, the correctness of \refalg{buniv} relies on \refprop{recycle}.
Below, we state that the recycled consensus objects used in \refalg{buniv} satisfy this proposition,
then we reduce \refalg{buniv} to \refalg{runiv}.

\begin{proposition}
  \labprop{buniv:1}
  During every execution \run of \refalg{buniv}, for every natural $k$, $\recycle{\ifunctOf{k}}$ implements a recycled consensus object.
\end{proposition}

\begin{proof}
  Choose an execution $\run$ of \refalg{buniv} and suppose that P1 holds in $\run$ for all objects $\recycle{\ifunctOf{k}}$, up to some operation $\op$ on $\recycle{o,t}$.
  The case $t=0$ is obvious, hence from now we assume $t>0$.
  Consider an operation $\op'$ on $\recycle{o,t'}$ such that $\op'$ does not precede $\op$ in $\run$ and $t'<t$ holds.
  Let $p$ be the process that executes $\op$ in $\run$.

  Operation $\op$ on $\recycle{o,t}$ occurs either at \refline{buniv:9} or at \refline{buniv:16}.
  Since $t>0$, process $p$ executed a call to function \enter{} at \refline{buniv:7} before, 
  and this call returned $\recycle{\ifunctOf{l},t}$, with $o=\ifunctOf{l}$.
  From the code at \refline{buniv:11}, naturals $l$ and $t$ are both contained in the decision register of some object $\recycle{\msgAny,t-1}$.
  Then because P1 applies up to $\op$, some process $p''$ invoked $op''=\propose{(\msgAny,\msgAny,l,t)}$ on $\recycle{\msgAny,t-1}$ previously.
  At this point, we may consider the two following cases.
  (Case $t'=t-1$)
  The index $l$ computed by $p''$ is the result of a call to function \free at \refline{buniv:16}.
  But from the code at \reflines{buniv:1}{buniv:5}, at that time $L[p'']$ equals $l$.
  A contradiction.
  (Case $t'<t-1$)
  We apply our induction hypothesis to $\op''$ and $\op'$.
\end{proof}  

\begin{proposition}
  \labprop{buniv:2}
  During every execution $\run$ of \refalg{buniv}, $h|\enter{}$ is a correct history for a racing on consensus objects.
\end{proposition}

\begin{proof}
  The domain of $\enter{}$ consists of all the objects $\recycle{o,t}$ with $o=\ifunctOf{l}$ and $t \in \timestampSet$.
  For some execution $\run$ of \refalg{buniv}, we define $\ll_{\run}$ as the order induced by $<$ on $\timestampSet$, that is 
  $\recycle{o,t} \ll_{\run} \recycle{o',t'}$ holds iff $t<t'$.

  By \refprop{buniv:1}, $\recycle{o}$ is a recycled consensus object.
  Thus, we may consider a decomposition of $\run|\recycle{o}$ in rounds $\{r_1, \ldots, r_{m \geq 1}\}$, where each round $r_i$ is a correct history for consensus.
  Choose some object $\recycle{o,t}$.
  When some process $p$ executes an operation $\op$ on $\recycle{o,t}$, we may consider the round $r_i$ to which operation $\op$ belongs.
  Because $r_i$ is a correct history for a consensus object, this is also the case for the history $\run|\recycle{o,t}$.
  This proves that the domain of $\enter{}$ is a set of consensus objects.

  Then, we observe that when $p$ accesses $c_i=\recycle{\ifunctOf{l},t}$ with some operation $\op$, $p$ executed a call to $\enter{}$ previously.
  From the code at \refline{buniv:11} and the code of function $\enter{}$, the pair $(l,t)$ is the result of a decision stored in $\recycle{o',t-1}$.
  Hence, some process entered $\recycle{o',t-1}$ previously and this object is the greatest object smaller than $\recycle{o,t}$ for $\ll_{\run}$.
\end{proof}

\noindent
We can now state the main result of this section:

\begin{theorem}
  \labtheo{buniv}
  \refalg{buniv} implements an obstruction-free linearizable universal construction.
\end{theorem}

\begin{proof}
  From \refprop{buniv:2} and the code of \refalg{buniv}, \refalg{buniv} implements \refalg{runiv}.
  Applying \reftheo{runiv}, we deduce that \refalg{buniv} implements an obstruction-free linearizable universal construction.
\end{proof}